\documentclass[letterpaper,twocolumn,10pt]{article}
\usepackage{usenix2019_v3}

\usepackage{tikz}
\usepackage{amsmath}

\usepackage[utf8]{inputenc}
\usepackage{algpseudocode}
\usepackage{algorithm}
\usepackage[english]{babel}
\usepackage{amsthm}
\usepackage{amsfonts}
\newtheorem{theorem}{Theorem}

\usepackage{xcolor}
\usepackage{xspace}
\usepackage{graphicx}

\usepackage{filecontents}

\begin{document}

\date{}

\title{Privacy-Preserving Methods for Outlier-Resistant Average Consensus and Shallow Ranked Vote Leader Election}

\author{
\rm Luke Sperling\\
\rm Computer Science and Engineering,\\
\rm Michigan State University\\
\rm East Lansing, USA\\
\rm sperli14@msu.edu
\and
\rm Sandeep S Kulkarni\\
\rm Computer Science and Engineering,\\
\rm Michigan State University\\
\rm East Lansing, USA\\
\rm sandeep@msu.edu
} 



\newcommand{\publickey}{\ensuremath{key_p}\xspace}
\newcommand{\secretkey}{\ensuremath{key_s}\xspace}


\maketitle

\begin{abstract}
    Consensus and leader election are fundamental problems in distributed systems.
    Consensus is the problem in which all processes in a distributed computation must agree on some value. Average consensus is a popular form of consensus, where the agreed upon value is the average of the initial values of all the processes. 
    In a typical solution for consensus, each process learns the value of others' to determine the final decision. However, this is undesirable if processes want to keep their values secret from others. 
    
    With this motivation, we present a solution to privacy-preserving average consensus, where no process can learn the initial value of any other process. Additionally, we augment our approach to provide outlier resistance, where extreme values are not included in the average calculation. Privacy is fully preserved at every stage, including preventing any process from learning the identities of processes that hold outlier values. To our knowledge, this is the first privacy-preserving average consensus algorithm featuring outlier resistance.
    
    In the context of leader election, each process votes for the one that it wants to be the leader. The goal is to ensure that the leader is elected in such a way that each vote remains secret and the sum of votes remain secret during the election. Only the final vote tally is available to all processes. This ensures that processes that vote early are not able to influence the votes of other processes. We augment our approach with shallow ranked voting by allowing processes to not only vote for a single process, but to designate a secondary process to vote towards in the event that their primary vote's candidate does not win the election.
    
\end{abstract}

\section{Introduction \label{intro}}

This paper focuses on two fundamental problems in distributed computing, consensus and leader election, and presents algorithms that preserve privacy while solving these problems.

Consensus \cite{DBLP:conf/opodis/Lamport02} is a fundamental problem in distributed computation. In consensus, all processes participating in a computation must agree on a shared value. This has applications in many scenarios, such as clock synchronization, leader election, and cloud computing.

In this paper, we focus on the case where the initial input is real-valued. And, the goal is that all processes decide on the average of those inputs (possibly, excluding some outliers).

Leader election \cite{DBLP:journals/tc/Garcia-Molina82} is another fundamental problem in distributed computation. The goal is for each process to eventually decide on a single process it thinks/wants as the leader. In the end, all processes should agree on which process is the leader.

A method of performing leader election is via ballot-casting, where each process designates a single other process it wishes to elect. Whichever process receives the most votes is elected as the leader. There needs to be a method of tie-breaking in the case that two processes receive the plurality of votes.

One concern in consensus and leader election is privacy. Traditional consensus assumes that votes are public. There are many reasons for wanting to keep votes private. In practical applications, these initial states may represent sensitive information. Similarly, in the leader election where ballot-casting is used, the votes should be kept private as well.

Situations where accurate reporting is desired but the information being reported is sensitive are concrete applications for privacy-preserving consensus algorithms. For example, the United States government may wish to collect data from large tech companies regarding how many security breaches they've faced in the past year. To encourage accurate reporting of this information, privacy-preserving solutions may be employed where each entity would be guaranteed that their data will remain private from all other entities as well as the government and the government can only learn the average value.
Another application of privacy-preserving consensus is the multi-agent rendezvous problem, where multiple agents wish to agree on a location to meet but do not want to disclose their initial location \cite{lin2004multi}. 

There are two types of average consensus. In the first type,  the goal is for the participant processes to compute the average for themselves. And, the goal is to prevent participant processes from learning others' contributions. In the second type,  we have a trusted/collector process that is interested in computing the average. And, the average value should be known only to this trusted process. However, none (including the trusted process) should learn the original values of other processes. 

In this paper, we introduce new algorithms that solve privacy-preserving average consensus and privacy-preserving leader election. Our approach includes a solution to privacy-preserving outlier-resistant consensus, which is a previously open problem, as far as we can tell. We leverage the properties of Homomorphic Encryption (HE) that allow computations over encrypted data without access to the secret key needed for decryption \cite{gentry2009fully}. Our leader election algorithm allows shallow ranked voting, such that each process not only may vote for their most-wanted candidate but also may indicate a secondary candidate to vote for if their primary candidate is not elected.

The intuition of our approach is as follows: 
In the absence of a trusted process, each process encrypts its initial state with its own public key. These ciphertexts are passed around the other processes and contributed to before being returned to the keyholder for decryption.
In the presence of a trusted process, that process holds exclusive access to the secret key. The other processes use the public key to encrypt their initial values that are homomorphically pooled together via message-passing. After aggregating all of the votes, the average may be calculated without decrypting the value. No process, not even the trusted process, learns the initial value of any process.

\textbf{Contributions of the paper. }
Our contributions in this paper are as follows:
\begin{itemize}
    \item We present a novel algorithm solving privacy-preserving average consensus in the presence of a trusted third party where no process, not even the trusted process, is able to determine the initial value of any other process.
    \item We modify the above algorithm for solving privacy-preserving average consensus with no trusted third party. Here, the initial state of each process is kept secret from each other process.
    \item We present a novel algorithm solving outlier-resistant privacy-preserving average consensus, which is a previously unsolved problem to our knowledge. Initial values which are considered outliers (determined as several standard deviations away from the mean) are not included in the mean calculation. Additionally, the identity of the processes holding extreme values cannot be deduced by any processes.
    \item We identify a  novel algorithm for solving privacy-preserving leader election. In this algorithm, the vote of each process is kept private and no process is able to gain any information that indicates which processes are more likely to win the election until the results are determined. Additionally, we provide the option for processes to designate a secondary vote for shallow-ranked voting.
\end{itemize}

\textbf{Organization of the paper. } 
Section \ref{background} contains necessary background on homomorphic encryption. Section \ref{problem} introduces and defines the system specifications. Sections \ref{approach1} and \ref{approach2} detail our solutions to average consensus and outlier-resistant average consensus, respectively. Section \ref{leaderelection} describes our privacy-preserving leader election algorithm. Section \ref{related} discusses related work in the literature. Finally, section \ref{conclusion} provides concluding remarks.

\section{Background - Homomorphic Encryption \label{background}}
To ensure the privacy of the initial values of the processes, we employ the Cheon-Kim-Kim-Song (CKKS) encryption scheme \cite{cheon2017homomorphic}. The CKKS scheme allows for approximate computations over encrypted values without access to the secret key needed for decryption. Its hardness is based on the Ring Learning with Errors problem. 
Entire real-valued vectors are encoded as plaintexts which are of the form $R = \mathbb{Z}[x]/(x^N + 1)$ before being encrypted via public key as ciphertexts. 
Ciphertexts are pairs of polynomial rings in the form $R_q^2 = \mathbb{Z}_q[x]/(x^N + 1)$ where $R_q$ represents polynomials of degree less than $N$ and coefficients modulo $q$.
The following operations are supported:

\textbf{Key Generation}: Given a set of parameters (such as level of security), generates a public key and secret key.

\textbf{Encryption}: Encrypt a plaintext into a ciphertext using the public key.

\textbf{Decryption}: From a ciphertext and the secret key, recover the original plaintext.

\textbf{Addition}: Two ciphertexts, or a ciphertext and a plaintext, are added together to result in a new ciphertext. This corresponds to the element-wise addition of the underlying vectors. Formally, $add(Enc[x_1,x_2,...,x_n],Enc[y_1,y_2,...,y_n]) = Enc[x_1+y_1,x_2+y_2,...,x_n+y_n]$.

\textbf{Multiplication}: Homomorphic multiplication translates to element-wise multiplication of the underlying vectors. Relinearization is needed after homomorphic multiplication. Formally, $mult(Enc[x_1,x_2,...,x_n],Enc[y_1,y_2,...,y_n]) = Enc[x_1y_1,x_2y_2,...,x_ny_n]$

\textbf{Relinearization}: Reduces the number of polynomials in a ciphertext from three to two to prevent the size of the ciphertext from growing from repeated multiplications.

\textbf{Rotation}: Using an optionally-generated set of rotation keys, ciphertexts may be cyclically rotated. In particular, it is possible to compute $Enc(x_2, x_3, \cdots, x_n, x_1)$ from $Enc(x_1, x_2, x_3, \cdots, x_n)$ without decryption. This function takes a parameter that identifies the order of rotation (left or right) as well as the amount of rotation. 

We note that this encryption scheme does not suffer from small-domain attacks. Specifically, encrypting the same string (say $0$) can generate multiple possible outputs. Thus, one cannot attack it by generating all ciphertexts when the domain of votes is small. 




\section{System Specifications \label{problem}}
\textbf{System Model:} We consider an asynchronous distributed system where processes communicate with one another via message-passing. We do not assume that all processes can send a message to any other process, but rather can only send messages to their neighbors. This relationship can be described by a communication graph $G=\{\Pi,E\}$ where $\Pi$ denotes the processes (acting as nodes of the graph) and $E\subseteq \{\{p_i,p_j\} | p_i,p_j\in \Pi,i\neq j\}$ denotes edges of the graph that identify the neighbor relation.
All communication between processes is assumed to be encrypted (with standard encryption techniques) with the receiver's public key to prevent any possibility of eavesdropping and causing a privacy violation. Although all sensitive information is homomorphically encrypted, this is done to prevent the homomorphic keyholder from eavesdropping to learn the private information (which would otherwise be a possible attack).

\textbf{Adversary Capability:} We assume that any adversaries are Honest-But-Curious, meaning they follow the protocol exactly but save a copy of any value they observe and try to deduce any information possible. In this work, we do not consider byzantine processes,
but rather focus on the aspect of privacy preservation.

\section{Privacy-Preserving Consensus \label{approach1}}

Our approach leverages the qualities of HE to provide strong privacy guarantees while still arriving at average consensus among all processes. No process is able to learn the starting value of any other process. We provide two algorithms to handle two different scenarios: the situation where a trusted third party is present and the situation where no process is trusted by all other processes.

\subsection{Problem Statement}
Each process $p_i$ holds an initial value $v_i \in \mathbb{R}$. By the end of the computation each process should decide on a value satisfying the following  conditions, where validity and agreement are changed to reflect the need for deciding on a final value that is an average of all votes, and a requirement for privacy is added. 

\begin{itemize}
    \item Validity and Agreement: The decided value is the average of the initial values of all $n$ processes. Formally, the decided value is $\frac{1}{n}\sum_{i=1}^{n} v_i$.
\end{itemize}
\begin{itemize}
   \item Termination: Every correct process eventually terminates.
    
    \item Privacy: No process is able to learn the initial value of any other process.

\end{itemize}

\subsection{Privacy-Preserving Consensus - Trusted Third Party}


The trusted party creates two homomorphic keys: a public key, \publickey, used to encrypt data and a secret key, \secretkey, for use in decrypting data. All processes in the computation have access to \publickey, but only the trusted process knows \secretkey.

In this algorithm, initially, a process creates two vectors, \textit{Votes} and \textit{Counts}.
\textit{Votes} represents the global state of the system and will have a vote from each process in its indices.  \textit{Counts}, represents how many times each process has voted. 
The trusted process aims to learn the average of the initial values of the other processes, but should not be able to learn the initial values of any other process. The trusted process alone holds \secretkey, the secret key needed to decrypt the \textit{Votes} ciphertexts of the other processes, ensuring that no other processes may learn the initial states of other processes.


After setup, process $i$ sends \textit{Votes}, where $Votes_i = v_i$ and $Votes_j = 0, j \neq i$. It also sends \textit{Counts}, where $Counts_i = 1$ and $Counts_j = 0, j \neq i$, 
\textit{Votes} is encrypted with \publickey and \textit{Counts} is sent as plaintext. This ensures that each process knows how many times each process has voted but it does not know its vote. Each process waits until it receives a message. Each time a message is received it sums the message's \textit{Votes} with its local \textit{Votes} (which translates to element-wise addition of the underlying vectors) and the message's \textit{Counts} with its local \textit{Counts} \footnote{
It is possible that a message provides no new information to the process. If a message's Counts' set of nonzero indices is a subset of the process' local Counts, then the message is ignored.
}. Formally, the vector underlying $Votes$ becomes $[Votes_1+m.Votes_1, Votes_2+m.Votes_2, ... , Votes_n+m.Votes_n]$.
It then broadcasts these updated variables to its neighbors. This procedure continues until the local Counts of a process contains no nonzero elements.


When Algorithm \ref{alg:trusted} terminates, \textit{Votes} is of the form $Enc([d_1 v_1, d_2 v_2, ..., d_n v_n])$ for some integer vector d.
and \textit{Counts} is of the form $[d_1, d_2,...d_n]$, where $\forall j, d_j > 0$


As shown in Figure \ref{fig:overview}, a process's vote may have been added to \textit{Votes} multiple times due to the protocol of each process adding each neighbor's \textit{Votes} value to its own. \textit{Counts} keeps track of how many times this has happened per vote to undo multiple voting in the next step.


We need to compute element-wise division of \textit{Votes} by \textit{Counts} so that vote of each process is counted exactly once. 
Unfortunately, Element-wise division is not supported by CKKS. Hence, we first compute $\frac{1}{d_1n}, \frac{1}{d_2n}, ...$.  This is possible since $d_1, d_2, ...$ are available as plaintext in $Counts$. 
Then, we multiply $[\frac{1}{d_1n}, \frac{1}{d_2n}, ... ]$ with $Votes$ which is of the form $Enc([d_1 v_1, d_2 v_2, ..., d_n v_n])$. This results in $Enc([\frac{v_1}{n}, \frac{v_2}{n}, ..., \frac{v_n}{n}])$.


Next, we need to sum up all the elements of this modified \textit{Votes}. By rotating the ciphertext left by one, the underlying vector becomes $[\frac{v_2}{n},\frac{v_3}{n},...,\frac{v_n}{n},\frac{v_1}{n}]$. Doing this $n$ times and adding them all together produces the sum of all the elements. This requires $n$ rotations and additions. But this can be achieved with $log(n)-1$ rotations and additions using Algorithm 1 in \cite{boddeti2018secure}.

Upon performing this computation, the underlying vector of $Votes$ will contain the average in each index. Now, this message can be communicated to the trusted party. The trusted party can then learn the average by decrypting the message. While other processes cannot decrypt this message, they can broadcast it to others so that each process will have access to the encrypted average value for further homomorphic computation.


    \begin{algorithm}
    \caption{Privacy-preserving average consensus}\label{alg:trusted}
    \begin{algorithmic}
    \Function{InitConsensus}{ }
    \State $Votes \gets Enc([0,...,v_i,...,0])$
    \State $Counts \gets [0,...,1,...,0]$
        \For{$p_j \in N_i$}
            \State $Send((Votes,Counts),p_j)$
        \EndFor
    \EndFunction
    
    \Function{Receive}{message m}
    \State $Votes \gets Votes + m.Votes$
    \State $\forall j : Counts_j \gets  Counts_j+ m.Counts_j$
    \For{$p_j \in N_i$}
        \State $Send((Votes,Counts),p_j)$
    \EndFor
    \If{$0 \notin Counts$}
        \State decide $Prepare(Votes, Counts)$
    \EndIf
    \EndFunction
    \Function{Prepare}{$Votes, Counts$}
    \State $Counts \gets [\frac{Counts_1^{-1}}{n},\frac{Counts_2^{-1}}{n},...,\frac{Counts_n^{-1}}{n}]$
    \State $Votes \gets mult(Votes,Counts)$
    \For{$i=log(n)-1$ to $0$} 
        \State $Votes \gets add(Votes,rotate(Votes,2^i))$
    \EndFor
    \State return $Votes$
    \EndFunction
    \end{algorithmic}
    \end{algorithm}


\begin{figure}[h]
    \centering
    \includegraphics[width=0.2\textwidth]{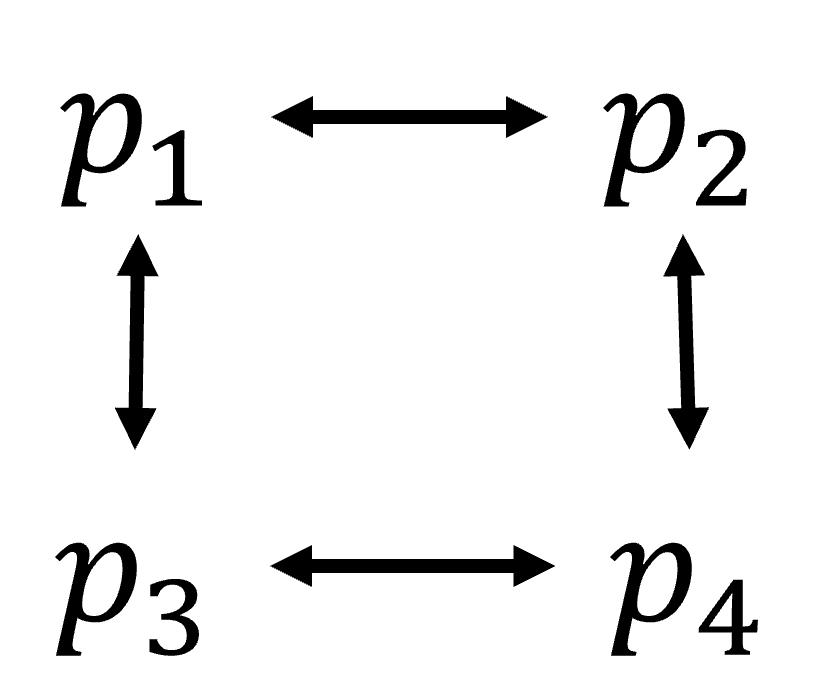}
    \includegraphics[width=0.45\textwidth]{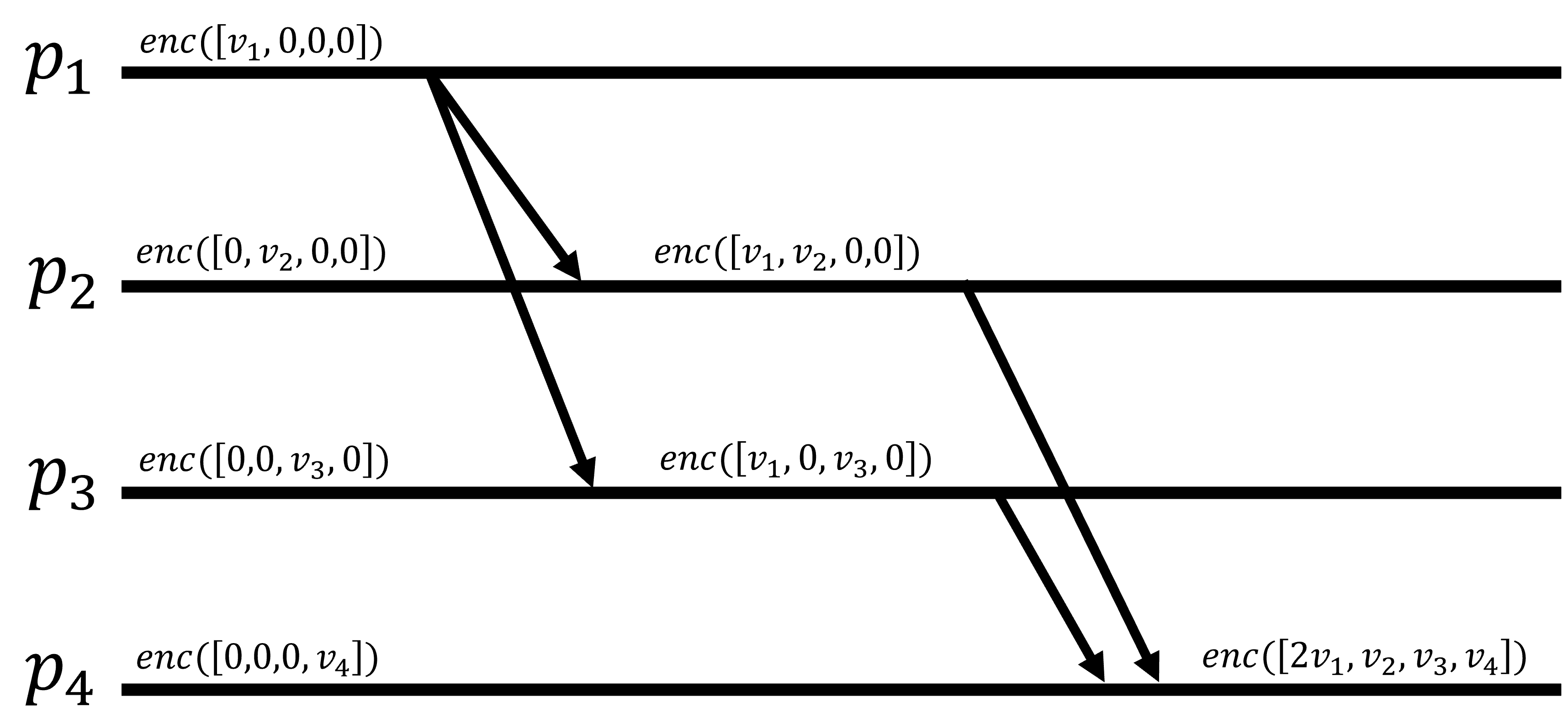}
    \caption{Time series chart of sample execution. Given the communication graph (above), a sample execution is shown (below). Processes 2 and 3 each send their local state to process 4, which can tell by the Counts portion of the message that the message contains new information. However, the vote of process 1 must be counted twice in order to incorporate all of this information. The Prepare phase resolves this issue.}
    \label{fig:overview}
\end{figure}

\subsection{Privacy-Preserving Consensus - No Trusted Parties}
In situations where no third party can be trusted, the previous algorithm can be modified to enable privacy-preserving average consensus without a trusted party. The main idea to achieve this is that (1) each process creates its own public/private homomorphic key, and (2) each process encrypts its own vote with \textit{its own} homomorphic public key and sends that to its neighbors. Algorithm \ref{alg:trusted} is then run with the initiator (keyholder) process being treated as the trusted process. This keyholder does not participate in the rest of the computation and is not sent any messages during the protocol until after the Prepare steps are taken. This is done to ensure no privacy breach can occur, as before the Prepare phase decrypting the ciphertext reveals the initial states of the other processes.

This protocol is done $n$ times concurrently, one for each process such that there are $n$ different public keys in use. At the end of the protocol, every process will learn the average but is unable to learn the initial state of any other process.

The above protocol will work correctly if the removal of the initiator does not partition the network. If the removal of the initiator partitions the network then no partition will have all the votes thereby preventing any process from calculating the average. However, if the original graph is connected, the protocol will succeed in computing the consensus value for some initiators. For example, if the underlying structure is a tree then the protocol will succeed when the leaves initiate the protocol. An initiator that succeeds can broadcast the average so others can learn it. 

The protocol can be easily revised so only select processes initiate the protocol. However, in this case, it would be necessary to select the relevant initiators and special care would be needed to deal with the case where the initiators fail. 


\subsection{Analysis }

In this section, we show that Algorithm \ref{alg:trusted} satisfies the desired properties of termination and privacy preservation. We also show correctness, i.e., we show that each process computes the average of the initial votes of processes. 
Finally, we discuss fault-tolerance aspects of Algorithm \ref{alg:trusted}. 

\subsubsection{Termination}
\begin{theorem} \label{theorem:term1}
If the communication graph G is connected and no process fails, all participating processes in Algorithm \ref{alg:trusted} terminate in finite time.
\end{theorem}
\begin{proof}
Processes following Algorithm \ref{alg:trusted} terminate when all other processes have shared their initial states with it. Information passes between neighbors each time a process sends messages. For information to travel to all processes in the graph, diameter(G) hops must occur. Therefore, after diameter(G) hops, the algorithm terminates.
\end{proof}

\subsubsection{Correctness}

\begin{theorem}
If the communication graph G is connected, all participating processes in Algorithm \ref{alg:trusted} decide on $\frac{1}{n} \sum_{i=1}^{n} v_i$. 
\end{theorem}
\begin{proof}
The value of the vector encrypted as $Votes$ is a linear combination of the initial Votes variables from each process. As such, the vector underlying $Votes$ can be expressed as $[d_1 v_1, d_2 v_2, ..., d_n v_n]$ for some integer vector d. Because $Counts$ is calculated the same way as Votes but with each process' initial $Counts$ vector, $Counts$ can be expressed as $[d_1, d_2, ..., d_n]$ for the same integer vector d. Performing element-wise division of $Votes$ by $Counts$ yields $[v_1, v_2, ..., v_n]$. Dividing each element by n and summing all elements yields the result.
\end{proof}

\subsubsection{Complexity}
A single message must make diameter(G) hops before Algorithm \ref{alg:trusted} terminates. This represents information from all processes transferring to all other processes. The communication complexity is therefore $O(diameter(G)|N_i|p)$ per process i, with p denoting the size in memory of a single ciphertext. Similarly, each process performs $O(diameter(G)|N_i|)$ homomorphic additions. The Prepare phase performs $O(n)$ inversions, $O(1)$ ciphertext-plaintext multiplications, and $O(log(n))$ ciphertext rotations and ciphertext additions. 
The version of the algorithm with no trusted parties is similar in both communication and computational complexity, although messages must return to the keyholder once Prepared. This results in a communication complexity of $O(2diameter(G)|N_i|p) = O(diameter(G)|N_i|p)$.

\subsubsection{Privacy}
\begin{theorem}
If all processes follow Algorithm \ref{alg:trusted}, no process is able to learn the initial value of any other process.
\end{theorem}
\begin{proof}
The only way to learn the initial value of another process in this algorithm is to obtain decryption of the Votes ciphertext before being put through the Prepare steps of rotation, addition, etc. However, no process other than the trusted process has access to the secret key needed for decryption. Furthermore, the trusted process has no access to the Votes ciphertext before the prepare phase, so not even the trusted process may violate the privacy of the other processes.
\end{proof}



\subsubsection{Fault Tolerance}
\begin{theorem}
Algorithm \ref{alg:trusted}  terminates under any number of detectable process faults, as long as the graph remains connected from the removal of any subset of faulty processes. 

\end{theorem}
\begin{proof}
Once the information from a faulty process reaches a correct process that is connected to every other correct process, the algorithm will terminate due to Theorem 1. This is true for any number of faulty processes as long as the condition holds that they deliver a message to a correct process connected to every other correct process before faulting. 
\end{proof}

    
    
    

\section{Outlier Resistant Privacy-Preserving Consensus\label{approach2}}
In this section we modify the previous algorithms to allow for processes whose initial value is considered an outlier to be excluded from the average calculation. Our algorithm neither reveals the initial value of any process nor does it reveal the identities of the processes that hold the outlier values.

\subsection{Modified Problem Statement}
Outlier-resistant average consensus as a problem is very similar to average consensus with the following two conditions added:
\begin{itemize}
    \item Outlier Resistance: Initial values that are deemed by some criteria to be outliers are not included in the calculation of the mean.
    \item Outlier Privacy: No process learns which other processes' initial values are outliers.
\end{itemize}

Outliers tend to be values that deviate substantially from the mean.
We define outliers to be any values that lie outside $\mu \pm c \sigma$, where $c$ is the parameter to the algorithm, $\mu$ is the mean and $\sigma$ is the standard deviation. 


\subsection{Outlier Resistant Algorithm}
The basic idea of the algorithm is to perform the Algorithm \ref{alg:trusted} three times: once to calculate the mean, once to calculate the standard deviation, and once to calculate the mean without outliers. These must be done in this order to protect privacy. The standard deviation is needed to determine if a process' initial value is an outlier. The mean is needed to determine the standard deviation.  Mean and standard deviation cannot be calculated together using this method, leading to the necessity of three rounds.

The first two rounds, computing the mean and standard deviation, use the same approach used in Algorithm \ref{alg:trusted}. In the first phase, a process uses $v_i$ as its initial value in Algorithm \ref{alg:trusted}. In the second phase, it uses $(v_i-mean)^2$ as its initial value to compute the standard deviation. 

There are two ways to do this. One way is for the trusted process to decrypt the mean value from round 1 and send it to all processes. In this case, the second round will be identical to the first round except for the initial value is different. Another way does not require the participation of the trusted process between the first and second round. Specifically, the first round computes $Enc(mean, mean, \cdots)$. We can multiply this itself to compute $Enc(mean^2, mean^2, \cdots)$. Each process can also multiply this with $[0,0,0,2v_i, 0,0]$, i.e., a vector where $i^{th}$ entry is $2v_i$ and other entries are zero. Finally, processes can user Algorithm \ref{alg:trusted} to compute the mean of $v_i^2$. Linear combination of these three entries will allow us to compute the average of $(v_i - mean)^2$, which is the same as the sum of averages of $v_i^2$, $2v_i mean$, and $mean^2$. Thus, the standard deviation can also be computed without involving the trusted third party. The third phase, however, requires participation from the third party to identify the bounds necessary to determine which processes are outliers. We describe the third phase, next. 


In the third round, conceptually, we run two concurrent instances of algorithm \ref{alg:trusted}. The first instance uses \textit{Votes} and \textit{Counts}, where a process votes $v_i$ if it is not an outlier and votes $0$ if it is an outlier. The second instance uses \textit{Participating} and \textit{Counts}. (The \textit{Counts} value is shared and needs to be included only once).
For the first instance, the Algorithm \ref{alg:outlier_resistant} will compute $\frac{\Sigma_{i \not \in Outliers} v_i}{n}$, where $n$ is the number of processes. The second instance will compute $\frac{\Sigma_{i \not \in Outliers}\  1}{n}$. The ratio of these numbers will provide the average without outliers. 
Here, process $i$ uses the initial value of $1$ if it is not an outlier and $0$ if it is an outlier for its local value of \textit{Participating}. 
\textit{Counts} is a plaintext whereas \textit{Votes} and \textit{Participating} are encrypted.


\subsection{Fault Tolerance}
Detectable fault tolerance can be added by adding the rule that in between each of the three rounds, all processes adjust the value of $n$ to be the current number of correct processes. Additionally, instead of terminating each round when \textit{Counts} has no nonzero elements, it terminates when \textit{Counts} has a nonzero element in every index corresponding to a correct process.

\subsection{Analysis}

In this section, we discuss the properties of Algorithm \ref{alg:outlier_resistant}. Namely, we prove theorems related to termination, correctness, and privacy. We also analyze the algorithm for communication and computational complexity.

\subsubsection{Termination}
\begin{theorem}
If the communication graph G is connected and no process faults, all participating processes in Algorithm \ref{alg:outlier_resistant} terminate in finite time.
\end{theorem}
\begin{proof}
Algorithm \ref{alg:outlier_resistant} runs Algorithm \ref{alg:trusted} three times. Theorem \ref{theorem:term1} proves that all three iterations terminate. Iteration 3 is modified but has the same termination condition.
\end{proof}

\subsubsection{Correctness}

\begin{theorem}
If the communication graph G is connected, all participating processes in Algorithm \ref{alg:outlier_resistant} decide on $\frac{\Sigma_{i \not \in Outliers} v_i}{\Sigma_{i \not \in Outliers}\  1}$. 
\end{theorem}
\begin{proof}
In Algorithm \ref{alg:outlier_resistant}, two values are computed concurrently in iteration 3: $\frac{\Sigma_{i \not \in Outliers} v_i}{n}$ from \textit{Votes}, and $\frac{\Sigma_{i \not \in Outliers}\  1}{n}$ from \textit{Participating}. The ratio of these numbers will provide the average without outliers. 
\end{proof}

\subsubsection{Complexity}
Because Algorithm \ref{alg:outlier_resistant} repeats the previous algorithm three times, its communication and computational complexity remain the same as well.

\subsubsection{Privacy}
\begin{theorem}
If all processes follow Algorithm \ref{alg:outlier_resistant}, no process can learn which processes hold initial values that are considered outliers.
\end{theorem}
\begin{proof}
Decrypting \textit{Votes} or \textit{Participating} prior to the Prepare phase reveals which processes hold outlier values. Similar to the previous algorithms, processes that have access to these variables have no access to the secret key, and the process that holds the secret key has no access to the variables before the Prepare phase is complete.
\end{proof}

    \begin{algorithm}
    \caption{Privacy-preserving outlier-resistant consensus}\label{alg:outlier_resistant}
    \begin{algorithmic}
    \Function{Main}{ }
    \State $InitialValue \gets v_k$
    \State $\mu \gets Consensus()$
    \State $InitialValue \gets (v_k - \mu)^2$
    \State $\sigma \gets Consensus()$
    \If{$|v_k-c\mu| > \sigma$}
        \State $Outlier \gets True$
    \Else
    \State $Outlier \gets False$
    \EndIf
    \State $\mu \gets OutlierResistantConsensus(Outlier)$
    \State Return $\mu$
    \EndFunction
    \Function{OutlierResistantConsensus}{$Bool$ $Outlier$}
    
    \State $Votes \gets Enc([0,...,v_i,...,0])$
    \State $Counts \gets [0,...,1,...,0]$
    \If{$Outlier$}
    \State $Participating \gets Enc([0,...,1,...,0])$
    \Else
    \State $Participating \gets Enc([0,...,0])$
    \EndIf
        \For{$p_j \in N_i$}
            \State $Send((Votes,Counts,Participating),p_j)$
        \EndFor
    \EndFunction
    
    \Function{Receive}{message m}
    \State $Votes \gets Votes + m.Votes$
    \State $\forall j : Counts_j \gets  Counts_j+ m.Counts_j$
    \State $Participating \gets Participating + m.Participating$
    \For{$p_j \in N_i$}
         \State $Send((Votes,Counts,Participating),p_j)$
    \EndFor
    \If{$0 \notin Counts$}
        \State decide $Prepare(Votes, Counts, Participating)$
    \EndIf
    \EndFunction
    
    \Function{Prepare}{$Votes, Counts, Participating$}
    \State $Counts \gets [\frac{Counts_1^{-1}}{n},\frac{Counts_2^{-1}}{n},...,\frac{Counts_n^{-1}}{n}]$
    \State $Votes \gets mult(Votes,Counts)$
    \State $Participating \gets mult(Participating,Counts)$
    \For{$i=log(n)-1$ to $0$} 
        \State $Votes \gets add(Votes,rotate(Votes,2^i))$
        \State $Temp \gets rotate(Participating,2^i)$
        \State $Participating \gets add(Participating,Temp)$
    \EndFor
    \State return $Votes$
    \EndFunction
    
    \end{algorithmic}
    \end{algorithm}

\section{Privacy-Preserving Leader Election \label{leaderelection}}

In this section we discuss privacy-preserving leader election using a similar approach to the previous sections. Our approach prevents any process from learning the vote of any other process and introduces novel approach for handling tie-breaking while preserving privacy. Specifically, we introduce ranked voting where, if a vote's primary choice is eliminated from the election, the vote counts for a secondary choice. This is achieved in a way that no process learns the primary or secondary votes of other processes. Additionally, no process learns which process is likely to win the election prior to a leader being elected. Specifically, the identities of the processes that are leading in with first votes is kept secret. This ensures that the secondary votes do not depend upon the identities of the leaders after considering the first vote. 

\subsection{Problem Statement}
In classical leader election, each process must decide if it is a leader or not. The problem is solved when exactly one process decides that it is the leader and all other processes know the identify of this leader. We also add two conditions based on privacy.
\begin{itemize}
\item Validity: The process with the most votes win. (This requirement can be fine-tuned. We consider two versions of validity. We also discuss other versions of validity requirement.)
    \begin{itemize}
        \item Each process casts only a primary ballot. A process that receives the maximum votes wins. Tiebreaks are \textit{independent} of process ID to ensure candidate privacy. This ensures that all processes that receive the same number of votes have an equal chance to be the leader. 
        \item Each process casts a primary and secondary ballot. If no process receives a majority with primary ballots alone, the secondary ballots are used to determine the winner. 
    \end{itemize}
    \item Uniqueness: Only a single process is selected as the leader. 
    \item Agreement: All processes agree on which process is the leader.
    \item Termination: Every correct process eventually terminates.
    \item Voter Privacy: No process learns the initial state (the vote) of any other process.
    \item Candidate Privacy: No process is able to learn any information that can help in guessing which process may be elected until a leader is determined.

\end{itemize}
An adversary that tries to make other processes fault may attempt to target a process that is likely to become leader in order to disrupt the computation. Candidate privacy prevents adversaries from knowing which process to target until the computation is complete.

In our work, we rely on a trusted process that provides a homomorphic public/private key, $key_p$ and $key_s$, respectively. All communication during the election is encrypted with $key_p$ (in addition to public key of the receiver). At the end, this message is given to the keyholder to reveal the identify of the leader. 

\begin{figure}[h]
    \centering
    \includegraphics[width=0.4\textwidth]{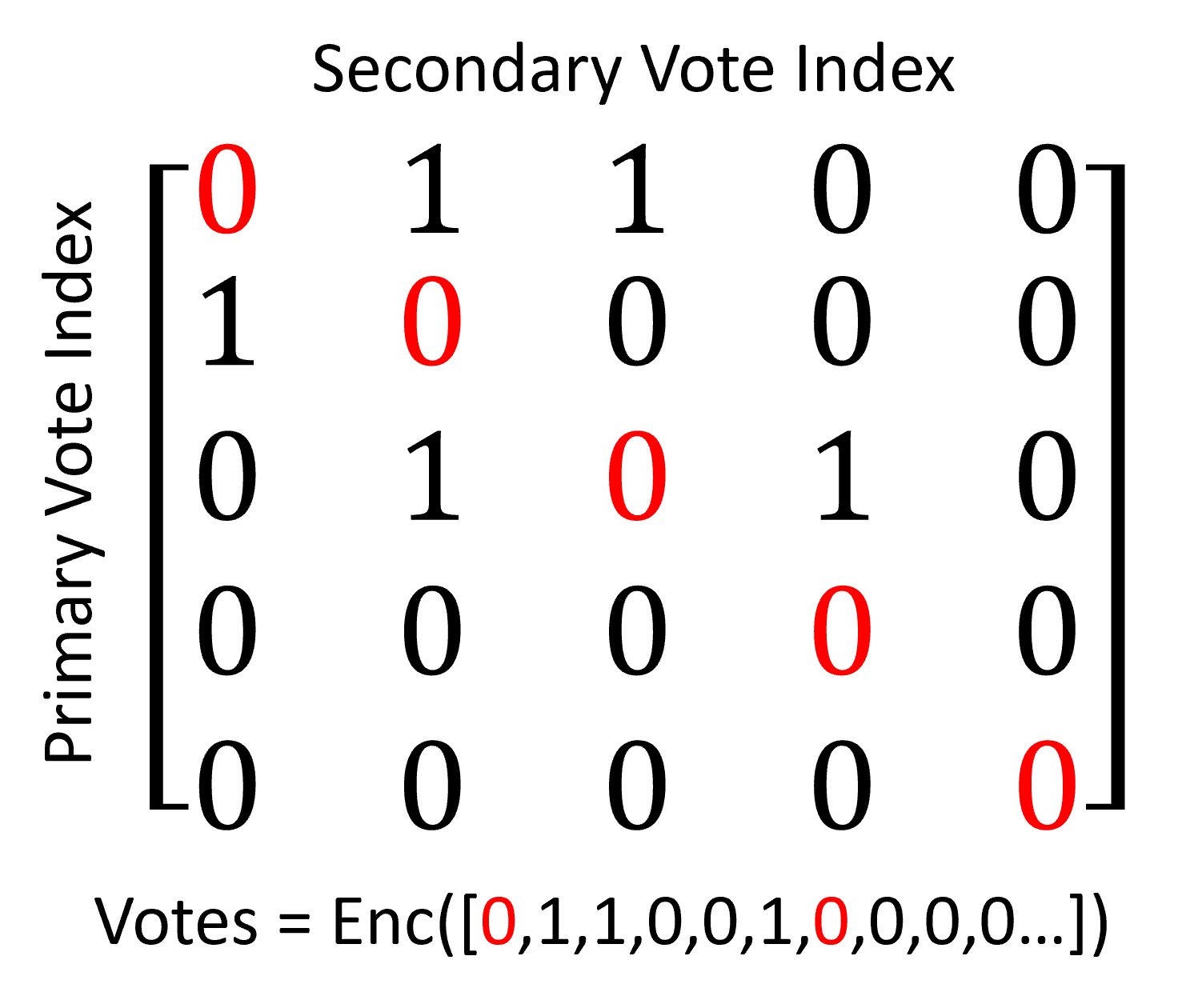}
    \caption{Privacy-preserving ballot-casting with ranked voting. When a process contributes to the \textit{Votes} ciphertext, it selects the index of the process it wishes to vote for (i) and the index of the process it designates as its secondary vote if its primary candidate cannot win the election (j). It then adds 1 to the index (i,j) in the matrix. Indices shown in red must remain zero, as the primary and secondary vote of a process cannot designate the same candidate. In order to encrypt this as a ciphertext, the matrix is flattened by appending each row of the matrix end-to-end. In this example, process 0 wins the election due to the secondary vote of the vote that initially counted towards process 1 transferring to process 0 to break the tie between processes 0 and 2.}
    \label{fig:ballot}
\end{figure}

\subsection{Privacy-Preserving Ballot-Casting}

An approach similar to that of privacy-preserving consensus can be applied to solve leader election. We need to make two changes to Algorithm \ref{alg:trusted}.

The first change is that each process creates its initial \textit{Votes} ciphertext to represent which process it wants to elect leader.
Specifically, if process $i$ wants to vote for process $j$, it creates the vector underlying $Votes$ to be $[0,...,1,...,0]$, where only $j^{th}$ entry is nonzero.
It creates a \textit{Counts} vector to be the same as the \textit{Counts} vector from the previous algorithms in that the $i^{th}$ entry is set to $1$ and all other entries are set to $0$. 
Thus, the \textit{Counts} vector indicates that process $i$ has voted whereas $Votes$ indicates that there is one vote for $j$. 
Similar to Algorithm \ref{alg:trusted}, \textit{Votes} is encrypted and \textit{Counts} is in plaintext. 

The second change is the removal of the Prepare phase. When all processes have contributed, decryption of the ciphertext reveals which process is chosen to be leader. Before this time, no process has any way of knowing which process will be elected. Additionally, no process can learn the vote of any other process, even after decryption. This is because each index of the encrypted vector represents how many votes that process received rather than the initial state of that process. In other words, no additional preparation must be done on a ciphertext after all votes are accrued in order to preserve the privacy of the participating processes. It just needs to be decrypted by the keyholder. 

An important difference in this new algorithm for leader election is the inability to aggregate \textit{Votes} ciphertexts of the same key. Instead, processes must add their vote to each ciphertext they receive, and must only add their vote to a ciphertext a single time. In other words,, a process adds its $Vote$ only if the $Counts$ vector indicates that it has not voted before.

\subsection{Breaking Ties with Ranked Voting}
A problem when performing leader election is what to do in case of a tie. A popular method is to elect the process with a higher ID number, but that reveals that processes with higher IDs are more likely to be elected than lower IDs, violating candidate privacy.

Ranked voting is an ideal solution for tie-breaking, however it is difficult to implement in a privacy-preserving way. Ranked voting calls for each vote to be tracked individually. This may result in a loss of privacy if each vote is its own ciphertext, as processes may figure out which ciphertext originated from which process. Additionally, ciphertexts are large in memory size, causing communication complexity to become very large if every vote were to be a ciphertext.

We introduce "shallow" ranked voting by allowing each process to submit a secondary, vote, but not rank every process. We use a secondary vote for use in the case of a tie. Each process not only casts a vote for their primary candidate, but additionally indicates a candidate that should receive the process's vote in the event that their primary candidate doesn't win the election. We achieve this with a secondary vote matrix. This matrix contains the secondary votes of each process, keeping track of which process was selected as the primary vote and the secondary vote, while keeping private the identity of the process casting the ballot. The way the leader is determined is as follows: If a candidate has the majority of the votes, they are elected. The candidate with the fewest votes is eliminated. Any votes that initially went to this candidate now are transferred the the voters' second choice for leader (using the matrix). This process repeats with the new candidate with the fewest votes until a single process remains. In traditional ranked voting, ballots rank each candidate but that requires $n^n$ vector size. Our approach simply uses a secondary vote, meaning that if a voter's primary and secondary candidates are eliminated, that vote will have no further bearing on the election.

Although only vectors may be encoded and encrypted homomorphically, matrices may be encoded by encoding the vector containing the concatenation of the rows of the matrix. An example of a secondary vote matrix is shown in Figure \ref{fig:ballot}.

\textbf{Breaking a tie. } 
If a tie still exists after the ranked voting, it may be pseudo-randomly broken independent of process ID. Consider $k$ processes tie with some number of votes $v_{tie}$. To break the tie, we compute $v_{tie} \bmod k$, which yields a number between $0$ and $k-1$ (inclusive). Suppose this number is $r$ then we can choose the process with $r^{th}$ ID to be elected as the leader. This allows all processes with highest number of votes to be elected as the leader. 


\subsection{Analysis}
In this section we prove termination, uniqueness, agreement, and privacy for our leader election protocol.

\subsubsection{Termination}
\begin{theorem}
If the communication graph G is connected and no process faults, all participating processes in our leader election algorithm terminate in finite time.
\end{theorem}
\begin{proof}
This algorithm has the same termination condition as \ref{alg:trusted} and thus also terminates by Theorem \ref{theorem:term1}.
\end{proof}

\subsubsection{Uniqueness and Agreement}
\begin{theorem}
If all processes follow our leader election protocol, exactly one process will be selected as the leader and every process will agree which process is the leader.
\end{theorem}
\begin{proof}
The leader is determined upon the decrypting of the \textit{Votes} ciphertext. Our protocol breaks any ties that may result from multiple processes receiving the same number of votes. Therefore, only one process will win the election. The keyholder then communicates this information to every other process. This ensures agreement as well.
\end{proof}

\subsubsection{Privacy}
\begin{theorem}
If all processes follow our leader election protocol, no process is able to learn the initial vote of any other process.
\end{theorem}
\begin{proof}
The \textit{Votes} ciphertext, when decrypted, does not reveal which processes contributed which votes, only how many votes each process received. If the ciphertext isn't decrypted until all processes have contributed, no process may determine which process contributed which vote. This holds true in the algorithm and thus the theorem is proved.
\end{proof}

\section{Related Work \label{related}}
Consensus is a fundamental problem in distributed computation and as such has seen many works dedicated to solving it and its variations. Early work focused on solving consensus in multiple contexts such as in the presence of faults \cite{fischer1983consensus}. In a fully asynchronous setting, the presence of a even single undetectable fault was shown to produce the possibility of nontermination \cite{fischer1985impossibility}, leading to consensus algorithms assuming some level of faults being detectable \cite{chandra1996unreliable} or focuses on reducing synchrony needed \cite{dwork1988consensus} to achieve consensus.

Much recent work in the field has been dedicated to consensus in the context of blockchain \cite{xiao2020survey}. Other work has been focused on optimization of convex problems as a constraint of consensus \cite{shi2017distributed} \cite{qiu2016distributed}. Still, other work focuses on mitigating attacks during consensus \cite{lu2018distributed} \cite{feng2016distributed}.

Leader election is a fundamental problem in distributed systems from the early work in  \cite{le1977distributed}. Early papers focused on topics such as mobile ad hoc networks \cite{malpani2000leader}, link failure tolerance \cite{singh1996leader}, and broadcast networks \cite{brunekreef1996design}.
Methods of stable leader election and self-stabilizing leader election have been the topic of recent research \cite{chen2019self} \cite{doty2018stable} \cite{sudo2020loosely} \cite{altisen2017self}. Additional topics include space optimal solutions \cite{gkasieniec2018fast} \cite{berenbrink2020optimal} and applications into internet of things \cite{wu2021privacy} \cite{rahman2019leader}.

Privacy-related concerns have driven research on privacy-preserving solutions in distributed systems in recent years. Privacy-preserving maximum consensus has been solved by generating and transmitting random numbers before transmitting initial states to hide the actual initial state \cite{duan2015privacy} \cite{bouchra2019robust}. These approaches must spend time transmitting random values prior to communication that includes useful information, incurring a drawback of communication complexity and time spent transmitting false values.

Homomorphic encryption (HE) is popular for its ability to provide strong privacy guarantees in distributed systems. The problem of consensus has been solved using HE \cite{lazzeretti2014secure} \cite{ruan2017secure} \cite{hadjicostis2018privary}. These approaches focus on gossip communication and short-term use of HE, i.e. processes encrypt data, communicate with their neighbor which modifies the ciphertext before sending it back for decryption. Our work is more general in that it allows computation over an arbitrary network to compute consensus with removal of outliers as well as ranked leader election. 

Blockchain has seen privacy improvement through HE via secret leader election \cite{freitas2022homomorphic} and trustworthy random number generation \cite{nguyen2019scalable}. These algorithms, although designed for use with blockchain, are generic to distributed computations.

\section{Conclusion \label{conclusion}}

Consensus and leader election are two fundamental problems in distributed computing. While there are several algorithms for solving these problems \cite{DBLP:conf/opodis/Lamport02,DBLP:journals/tc/Garcia-Molina82}, they assume that the votes/preferences of each process can be known to others. This is undesirable where each participant wants to preserve the privacy of their own vote. 

We focused on the problem of average-consensus, where the goal to compute the average votes of individual processes. The goal is to ensure that the average is computed while keeping the original votes secret. We considered two variations of the consensus problem, one where there is a trusted party that wants to compute the average and one where such a party does not exist. 

In the first variation, we only assume that the trusted party creates the necessary homomorphic public/private keys and provides the public key to everyone. This party is not trusted with any of the votes and we ensure that the actual votes are kept private from this party as well as other participants. A typical use of this model would be one where an entity, e.g., government, wants to collect statistical data on vitally important characteristics but wants to overcome resistance from entities, e.g., private companies, from sharing data with others. A typical instance of this one where the government may want to collect information about the number of average security attacks but each company wants to keep the number of security attacks on them private. 

Our algorithm does not consider collusion between a user and trusted entity. This assumption is reasonable since a user, say $i$, must send their votes to some other user $j$. If $j$ is colluding with the trusted entity, privacy of $i$ could be violated. 
However, it would be possible address some aspect of collusion.  To achieve this, we observe that our algorithm can be easily extended to cases where the set of users are partitioned into groups where one user can be part of multiple groups. Here, the algorithm without the trusted process can be used in each group to compute the average (or sum) of all votes. This cumulative information can then be exchanged with other groups. Since our algorithm already permits a process to add its vote several times without affecting the average, this approach is feasible even if a user is part of multiple groups. With this approach, privacy of process $i$ will be preserved even if processes in other groups (i.e., groups where $i$ is not a part) collude with the trusted entity.

In our second variation, we considered the case where a group of processes want to compute the average-consensus among themselves. 
A typical example of this would be a group of employees anonymously submitting ratings of their workplace between one another in order to accurately gauge sentiment. 


We also addressed the problem of average-consensus where outliers are omitted from calculation while still preserving privacy. Eliminating outliers in analyzing the data is often necessary so that single outlier data items do not corrupt the conclusion. Our solution guarantees that none can learn about which processes are the outliers. 

As part of solving the problem of average consensus while eliminating outliers, we also compute the standard deviation of the original votes. This can be of use in various other applications as well. For instance, it may be desirable to learn the average number of security breaches each major tech company faced in the previous year, but companies may not desire to share this information. Additionally, if a single company has faced many more breaches than the others due to a security flaw, a more accurate picture may be formed by excluding the outlier. In this case, that company may be unwilling to share that they themselves are an outlier because it could imply a high number of security breaches.

We also developed a privacy-preserving algorithm for leader election. Here, the goal was not only to ensure that votes are kept private but to ensure that the identities of potential leader candidates is also private. This candidate privacy guarantees that none is able to discern any information during the voting process. Candidate privacy ensures that the early voters are not able to affect votes of late voters. It also prevents \textit{strategic} voting where a person chooses to vote for a less preferred candidate because it can cause their preferred candidate to win. Our algorithm also does not do tie-breaking on ID, as tie-breaking on ID causes some processes (e.g., processes with higher ID) to have a higher chance of being elected than others. 

\bibliographystyle{plain} 
\bibliography{main}

\end{document}